\documentclass[12pt]{article}
\usepackage{epsfig, amssymb, graphicx, amsmath} 
\usepackage{amsthm}
\usepackage{relsize}
\usepackage{comment}
\usepackage{mathcomp}
\usepackage[round-mode=places, round-integer-to-decimal, round-precision=2,
table-format = 1.2, 
table-number-alignment=center,
round-integer-to-decimal,
output-decimal-marker={,}
]{siunitx} 
\usepackage{booktabs}
\usepackage{enumitem}
\usepackage{eurosym}	
\usepackage{wrapfig}	
\usepackage{float}		
\usepackage{footnote}
\usepackage{subcaption}

\usepackage{thmtools}
\usepackage{bbm}		
\usepackage[title]{appendix}
\usepackage{adjustbox}
\usepackage{multirow}
\usepackage{hyperref}
\usepackage{multicol}	
\usepackage[authoryear]{natbib}
\bibpunct{(}{)}{,}{a}{}{,}

\newtheorem{theorem}{Theorem}

\newtheorem{proposition}[theorem]{Proposition}
\numberwithin{equation}{section}

\theoremstyle{definition}
\newtheorem{definition}[theorem]{Definition}
\theoremstyle{remark}

\numberwithin{theorem}{section}
\numberwithin{proposition}{section}
\numberwithin{lemma}{section}
\numberwithin{corollary}{section}
\numberwithin{definition}{section}
\numberwithin{remark}{section}
\numberwithin{example}{section}

\newcommand{\be}{\begin{equation}}
	\newcommand{\en}{\end{equation}}
\newcommand{\ben}{\begin{equation*}}
	\newcommand{\enn}{\end{equation*}}
\newcommand{\bea}{\begin{eqnarray}}
	\newcommand{\ena}{\end{eqnarray}}

\begin{document}
	
	\newlength\tindent
	\setlength{\tindent}{\parindent}
	\setlength{\parindent}{0pt}
	\renewcommand{\indent}{\hspace*{\tindent}}
	
	\begin{savenotes}
		\title{
			\bf{ 
				Is (independent) subordination relevant\\ in option pricing?
			
		}}
		\author{
			Michele Azzone$^*$ \& 
			Roberto Baviera$^*$ 
		}
		
		\maketitle
		
		\vspace*{0.11truein}
		\begin{tabular}{ll}
			$*$ & Politecnico di Milano, Department of Mathematics, Italy.\\
			
		\end{tabular}
	\end{savenotes}
	
	\vspace*{0.11truein}
	\begin{abstract}
		\noindent
	  \citet{monroe1978processes} demonstrates that any local semimartingale can be represented as a time-changed Brownian Motion (BM).  A natural question arises: does this representation theorem hold when the BM and the time-change are independent? We prove that a local semimartingale is not equivalent to a BM with a time-change that is independent from the BM.

	  Our result is obtained utilizing   a class of additive processes: the additive normal tempered stable (ATS). This class of processes exhibits an exceptional ability to accurately calibrate the equity volatility surface.
 We notice that  the sub-class of additive processes that can be obtained with an independent additive subordination  is incompatible with market data and shows significantly worse calibration performances than the ATS, especially on short time maturities. These results have been observed every business day in a semester on a dataset of S\&P 500 and EURO STOXX 50 options.
	\end{abstract}
	
	\vspace*{0.11truein}
	\vspace*{0.11truein}

		\vspace{2cm}
	\begin{multicols}{2}
		{\bf Address for correspondence:}\\
		{\bf Roberto Baviera}\\
		Department of Mathematics \\
		Politecnico di Milano\\
		32 p.zza Leonardo da Vinci \\ 
		I-20133 Milano, Italy \\
		Tel. +39-02-2399 4575\\
		roberto.baviera@polimi.it
		\columnbreak
		
		$ $\\
		{\bf  Michele Azzone}\\
	Department of Mathematics \\
	Politecnico di Milano\\
	32 p.zza Leonardo da Vinci \\ 
	I-20133 Milano, Italy \\
		Tel. +39-338-2464 527 \\
		michele.azzone@polimi.it\\
	\end{multicols}
	
	\section{Introduction}
	A class of additive processes, the additive normal tempered stable processes (ATS), has recently shown excellent calibration properties of the implied volatility surface of equity derivatives \citep{azzone2019additive}. This study investigates whether it is possible to obtain this process by (independent) additive subordination.
	The technique of embedding  a process in a Brownian motion (BM) is frequently employed  in the  literature to introduce new (tractable) models \citep[see e.g.,][Ch.4 and references therein]{Cont}.
	 \citet{monroe1978processes} proves a well known representation theorem: any local semimartingale can be represented as a time-changed  BM. Nonetheless, in general, the time-change and the BM are not independent: we prove that a local semimartingale is not equivalent to a BM with a time-change {\it independent} from the BM.
	
	\smallskip
	
 Specifically, our findings establish the nonexistence of {\it any pair} comprising a time change and an independent BM, where the resulting time-changed BM is equal in law to the ATS. Additive processes (an extension of L\'evy processes) are characterized by independent but not necessarily  stationary increments \citep[for a detailed description see][]{Sato}. In recent years, these processes emerge as a new frontier in the quantitative finance literature, with several contributions exploring their properties and applications to financial markets  \citep[see e.g.,][]{madan2020additive,carr2021additive,azzone2019additive,madan2023economics}. 
	
	\smallskip
	
	Time-change or subordination is a well-established technique for building statistically relevant models.  As is customary in the literature, we refer to a positive and non-decreasing process that is zero a.s. at time t=0 as a subordinator. Subordination finds most of its financial applications in the stationary L\'evy framework   \citep[see e.g.,][]{Cont}. Additive subordination is introduced by  \citet{galloway2008subordination} and  later formalized by \citet{mijatovic2010additive}. In particular, a subordinated BM has a well-known financial interpretability. The time-change models economic time: the more intense the market activity, the faster economic time runs compared to calendar time \citep[see e.g.,][]{madan1998variance,Geman2001Time}. 

	\smallskip
	
	It is worth noting  a standard assumption in the financial literature: the stochastic time-change and the subordinated process (e.g. the BM) are  independent.
	This independence is particularly useful in the applications: the independence between the BM and the subordinator enables straightforward  simulation of the process \citep[see e.g.,][Ch.6 and references therein]{Cont} and allows an efficient extension to the multivariate case \citep[see e.g.,][]{linders2016multivariate,luciano2016dependence}. Hereinafter, following  \citet{barndorff2006infinite}, when we refer to a subordinator, we intend a subordinator independent from the BM.
	
	\smallskip
	
	Three are the main contributions of this paper.
	
	 First, we prove that, in general, an additive process $\{f_t\}$  is not always equivalent to a subordinated BM  $\{W(aZ_t)+bZ_t+c_t\}$ with $\{Z_t\}$ a subordinator independent from the BM  $\{W(t)\}$.\footnote{We follow the notation of \citet{Sato} and indicate a stochastic process on $\mathbb{R}$, $\{X_t : t \geq 0\}$, with $\{X_t\}$. When we write $X_t$ without  brackets we are referring to the process at time $t$.} Thus, the Monroe representation theorem does not hold when we require the additional hypothesis that the BM and the time-change are independent.

	Second, we demonstrate that an additive process $\{f_t\}$ cannot be expressed as  $\{W(a_t  Z_t) + b_t  Z_t  + c_t\}$, where $a_t$ and $b_t$ are non-constant functions of time.
	
	Finally, we design a statistical test  to reject the null hypothesis that the additive process calibrated on market data can be written as a subordinated BM with the additive subordinator independent from the BM. The test is conducted on derivative prices of the two most liquid equity indexes: the S\&P 500 and the EURO STOXX 50.  The null hypothesis is rejected for all business days in a six-month time interval with p-values always
	 below 1$\tcperthousand$. 
		\smallskip
	
	The rest of the paper is organized as follows. In section \ref{sec:main}, we prove the main theoretical results: in general, an additive process is not always equivalent to a subordinated BM and  cannot be expressed as  $\{W(a_t  Z_t) + b_t  Z_t  + c_t\}$. In section \ref{section:num}, we describe the dataset, the calibration method and we implement the statistical test that  rejects the null hypothesis that the additive process calibrated on market data can be written as a subordinated BM. Section \ref{sec:conc} concludes. Moreover, in appendix \ref{app:calib}, we report additional calibration results; we enlighten why an additive process written as a subordinated BM does not replicate market data.

	\section{Can additive processes be represented with additive\\ subordinators?} \label{sec:main}
	In this section, we derive the  main theoretical results. In \textbf{Theorem \ref{theorem:no_subordination}}, we prove that, an additive process $\{f_t\}$ cannot be always represented as a subordinated BM $\{W(aZ_t)+bZ_t+c_t\}$ with  subordinator $\{Z_t\}$. The process $\{f_t\}$, which we consider as a counterexample, is relevant in the financial literature. Moreover, in \textbf{Theorem \ref{theorem: sub_diff}}, we demonstrate  that  $\{W(a_t  Z_t) + b_t  Z_t  + c_t\}$, where $a_t$ and $b_t$ are deterministic functions of time, is an additive process if and only if $a_t$ and $b_t$ are constant. In the next section, we show that the class of additive processes that accurately describes market data falls within this set of processes that cannot be represented with additive subordination.
	
	\smallskip
	
	In the following,  we consider the usual definition of subordinated BM 
	 \begin{equation}\label{eq:sub}
	 	\{W(a  Z_t) + b  Z_t  + c_t\,: t\geq 0\}\;\;,
	 \end{equation}
 where $a\in \mathbb{R}^+$ and $b\in \mathbb{R}$ are respectively  the squared volatility and the drift of the BM \citep[see e.g.,][p.82]{madan1998variance}.
We also consider   the standard definition of additive process on a probability space $(\Omega,{\cal F},\mathbb{P})$ for $t\in \mathbb{R}^+$
	 \citep[see e.g.,][Def.14.1, p.455]{Cont}.
	A c\'adl\'ag stochastic process on $\mathbb{R}$	$ \left\{X_t\right\}$, $X_0=0$ a.s.  is an additive process if and only if it has independent increments and  is continuous in probability. An additive process $\left\{X_t\right\}$ is fully characterized by a family of generating triplets $(A_t,\nu_t,\gamma_t)$, see \citet[][pp.38-39]{Sato}. The triplet consists of the diffusion term  $A_t$, the L\'evy measure  $\nu_t$,  and the drift term $\gamma_t$. 
	
	\smallskip
	
	As discussed in the introduction, in recent years, several market models based on additive processes have been developed for derivative pricing. Among these models, the class of additive normal tempered stable processes (ATS) emerges as a promising choice for modeling equity options. The ATS present excellent calibration properties --being parsimonious in terms of parameters-- and have the  correct short time behavior \citep[see][]{azzone2021short}.
	The ATS class has been introduced through its characteristic function \citep[cf.][eq.4]{azzone2019additive}: 
	\begin{equation}
		\mathbb{E}\left[e^{iuf_t}\right]={\cal L}_t \left( iu \left(\frac{1}{2}+\eta_t \right)\sigma_t^2+\frac{u^2\sigma^2_t}{2};\;k_t  ,\alpha\right)e^{iu\,\varphi_t\,t}\;\;, \label{laplace}
	\end{equation}
	where $ \sigma_t $, $k_t$ are continuous on $[0,\infty)$ and $ \eta_t $, $\varphi_t$ are continuous  on $(0,\infty)$
	with $ \sigma_t > 0$, $ k_t \geq  0$ and $\varphi_t \,t$ goes to zero as $t$ goes to zero. $ \ln {\cal L}_t$ is defined as \begin{equation*}
		\ln {\cal L}_t \left(u;\;k,\;\alpha\right) :=
		\begin{cases} 
			\displaystyle \frac{t}{k}
			\displaystyle \frac{1-\alpha}{\alpha}
			\left \{1-		\left(1+\frac{u \; k}{1-\alpha}\right)^\alpha \right \} & \mbox{if } \; 0< \alpha < 1 \\[4mm]
			\displaystyle -\frac{t}{k}
			\ln \left(1+u \; k\right)  & \mbox{if } \; \alpha = 0 \end{cases}\;\; ,
	\end{equation*}
	with $\alpha\in[0,1)$.  We observe that the characteristic function of the ATS process (\ref{laplace}) is of bounded variation.  The \citet{monroe1978processes}  representation theorem can be applied to the ATS because any additive process with a characteristic function of bounded variation over finite intervals is a semimartingale  \citep[see e.g.,][Th.4.14, p.106]{jacod}. 
	
	In the next theorem, we demonstrate that we cannot obtain  a class of additive processes (the ATS), for which the  representation theorem in \citet{monroe1978processes} holds, by time-changing a BM with {\it any} independent subordinator.
	\begin{theorem}
		\label{theorem:no_subordination}
		\noindent 
		An additive process is not equivalent to a subordinated BM.
	\end{theorem}
	\begin{proof}
		We prove the thesis with a counterexample. We demonstrate	that  it does not exist a subordinated BM $\{W(aZ_t)+bZ_t+c_t\}$ identical in law to the ATS, the additive process with characteristic function  (\ref{laplace}), if  $\eta_t$ is non-constant.
		
		We compute the characteristic function of a subordinated BM
		by conditioning to the filtration of the subordinator,
		
		\[	\mathbb{E}\left[e^{iuW(aZ_t)+b Z_t+c_t\,}\right]={\cal M}_t \left( iu\, b+\frac{u^2\,a}{2}\right)e^{iu \,c_t}\;\;, \]
		where ${\cal M}_t$ is the Laplace transform of the subordinator.  
		To obtain the same characteristic function of  (\ref{laplace}) the following relation should hold
		\begin{equation*}
			\ln {\cal M}_t \left(u\right) =
			\begin{cases} 
				\displaystyle \frac{{d}_t}{\alpha}
				\displaystyle 
				\left \{1-		\left(1+u\, e_t\right)^\alpha \right \} & \mbox{if } \; 0< \alpha < 1 \\[4mm]
				\displaystyle -d_t
				\ln \left(1+u \;e_t\right)  & \mbox{if } \; \alpha = 0 \end{cases}\;\; ,
		\end{equation*}
		where $d_t=\frac{t}{k_t}({1-\alpha})$, $e_tb=(1/2+\eta_t)\sigma^2_t\frac{k_t}{1-\alpha}$ and $e_ta=\frac{\sigma^2_t}{2}\frac{k_t}{1-\alpha}$.

		We show that, if  $\eta_t$ is not constant in $t$, we cannot match the characteristic function of the ATS. 
		By solving  $e_tb=(1/2+\eta_t)\sigma^2_t\frac{k_t}{1-\alpha}$ and $e_ta={\sigma^2_t}\frac{k_t}{1-\alpha}$ for $e_t$ we get
		
		\begin{equation*}
			\frac{{b}}{{a}}={(1/2+\eta_t)}\;\;,
		\end{equation*}
		which is absurd if $\eta_t$ is non-constant. The statement holds because we have verified the identity in law
	\end{proof}
	
	Let us point out the main result of   \textbf{Theorem \ref{theorem:no_subordination}}. We have presented a counterexample utilizing the ATS with a non-constant $\eta_t$, proving that --if  $\eta_t$ is non-constant-- it is impossible to construct an ATS as a subordinated BM. What does happen if $\eta_t$ is constant? In the following, we prove by construction that,  if $\eta_t$ is constant, it is possible to build the ATS as a subordinated BM.
	\smallskip
	
	To substantiate this claim on the ATS with constant $\eta_t$, we  introduce 
	an  additive subordinator.\footnote{We recall that the additive subordinator is an additive process which is a subordinator, see e.g., \citet{mijatovic2010additive}.} 
	We call this subordinator the additive tempered stable subordinator (TSS). It is a natural extension  of the L\'evy tempered stable subordinator  \citep[see e.g.,][p.127]{Cont} when considering time-dependent  $k_t$ and  $\sigma_t$.  We define below its family of  generating triplets.\footnote{	In line with the prevailing conventions in the literature \citep[see e.g.,][Ch.3]{Cont}, we identify $\nu_t$, a L\'evy measure  on $\mathbb{R}$,    with the L\'evy density $\nu_t(x)$ such that 
		$\int_B \nu_t(x) dx=\nu_t(B) \;\forall B \in \mathbb{B} (\mathbb{R})$ and  $B \subset \{ x : |x|>\epsilon >0 \}$.}

	\begin{definition} \label{definition:ATSSub}
		The TSS  $\left \{Z_t\right \}$ is characterized by the family of triplets $(0,{\cal V}_t,\Gamma_t)$
		\begin{align}
			\label{eq:TSS}
			\begin{cases}
				{\cal V}_t\left(x\right)&:=\displaystyle \frac{t\sigma^{2\alpha}_t}{\Gamma(1-\alpha)}\left({\frac{1-\alpha}{k_t}}\right)^{1-\alpha}\left(\frac{e^{-\left(1-\alpha\right) \; x/(\sigma^2_t k_t) }}{x^{1+\alpha}}\right)\mathbbm{1}_{x>0}\\[4mm]
				\Gamma_t&:= \displaystyle \int^1_0 x \; {\cal V}_t(x) \; dx\;\;, \\
			\end{cases}
		\end{align}
		where $\alpha \in [0,1)$.
		$\sigma_t$ and $k_t$ are positive continuous functions of time such that
		\begin{enumerate}
			\item \( \displaystyle t \; {\sigma^2_t} \)  \quad \quad\;\;  is \( o\left(1\right)   \) for small $t$; 
			\item \(\displaystyle  \frac{ t}{k_t^{1-\alpha}} \; \sigma^{2\alpha}_t \;\; \) is \( o\left(1\right) \)  for small $t$ and non-decreasing;
			\item \( \displaystyle \sigma^2_t \; k_t \,\)\; \quad \quad is non-decreasing.
		\end{enumerate}
		\label{defsub}
	\end{definition}
	In \textbf{ Definition \ref{definition:ATSSub}}, we have introduced a family of triplets; we need to demonstrate   the existence of an additive subordinator characterized by these triplets. In order to accomplish this task, we  first  
	derive certain sufficient conditions under which an additive process is an additive subordinator.   Then, we  employ these conditions to prove that the TSS is indeed an additive subordinator.
	\begin{proposition} 
		An additive process  $\left \{Z_t \right \}$  is an additive subordinator if its family of generating triplets is such that,  for every fixed time $t$,  
		$A_t=0$, $b_t:=\gamma_t -\int_{0\leq x \leq 1}{x  \; \nu_t (dx)}$ non-decreasing and $\nu_t$  such that i)  $\int_{\mathbb{R}}\left(|x|\wedge 1\right)\nu_t(dx)<\infty$, ii) $\nu_t((-\infty,0])=0$.
		Moreover,  an additive subordinator $\left \{Z_t \right \}$  has a characteristic function with exponent
		\begin{equation} 
			\ln{\mathbb{E}\left[e^{iuZ_t}\right]}=ib_t u+\int_{x>0}\left(e^{iux}-1 \right)\nu_t(x)dx\;\;.  \label{b}
		\end{equation}
		\label{p2.5}
	\end{proposition}
	\begin{proof}
		This proof extends to additive subordination a known property of L\'evy subordination \citep[see e.g.,][Prop.3.10, p.100]{Cont}.
		\smallskip
		
		First, we prove that, given the conditions on the characteristic triplets, equation (\ref{b}) holds. Define  $L_t(x):=\mathbbm{1}_{|x|\leq 1}x\, \nu_t(x)$ and $M_t(x):=\left(e^{iux}-1 \right)\nu_t(x)$.
		We have that
		\begin{equation*}
			\ln{\mathbb{E}\left[e^{iuZ_t}\right]}={
				i\gamma_t u+\int_{\mathbb{R}}\left(e^{iux}-1- \mathbbm{1}_{|x|\leq 1}iux\right){\nu}_t(x) \, dx}={
				i\gamma_t u+\int_{\mathbb{R}}\left(iuL_t(x)+M_t(x)\right)dx}\;\; .
		\end{equation*}
		The first equality is due to the definition of the characteristic function of an additive process  with no diffusion  \citep[cf.][Th.8.1, p.37]{Sato}. $L_t(x)$  is integrable with respect to  $x$	
		thanks to the conditions on $\nu_t $.
		The sum of $iuL_t(x)$ and $M_t(x)$ is integrable, because $\mathbb{E}\left[e^{iuZ_t}\right]$ is a well defined characteristic function, thus $M_t(x)$ is integrable too. 
		We can split the integral and check the thesis defining $b_t:=\gamma_t-\int_{0\leq x \leq 1}{x  \nu_t (dx)}$. This proves equation (\ref{b}).
		
		\smallskip
		
		Second,	 $\forall\;s,\;t$ such that $ 0\leq s<t$,  we prove that the increment $Z_t-Z_s$ is  a  non-negative random variable almost surely.
		By definition of additive process,  we get the characteristic function of $Z_t-Z_s$
		\[\mathbb{E}\left[e^{iu\left(Z_t-Z_s\right)}\right]=\mathbb{E}\left[e^{i u Z_t}\right]/\mathbb{E}\left[e^{i u Z_s}\right]\;\;\]
		and, using  (\ref{b}), we obtain an explicit formula for its exponent
		\[\ln \mathbb{E}\left[e^{iu\left(Z_t-Z_s\right)}\right] = i(b_t-b_s) u+\int_{x>0}\left(e^{iux}-1 \right)(\nu_t(x)-\nu_s(x))dx\;\;.\] 
		We observe that  $ b_t-b_s$ is non-negative, because $b_t$ is non-decreasing by hypothesis\footnote{We point out that if $s=0$ then $b_t\geq0\,\forall t$ because $b_t$ is non decreasing and $b_0=0$. The latter is true because for any additive process $\gamma_0=0$ and $\nu_0=0$.}, and that $\nu_t(x)-\nu_s(x)$ is a jump measure with non-negative jumps; $\nu_t(x)-\nu_s(x)$ is a non-negative function thanks to \citet[][Th.9.8, p.52]{Sato}  and with value on $[0,\infty)$ by hypothesis. Thus, the increment   $Z_t-Z_s$ is positive a.s. having  a positive drift and a positive jump measure.
		Summing up, a non-decreasing additive process with $Z_{t=0}$=0 a.s. is an additive subordinator  
	\end{proof}
	
	In the next proposition, we prove that the TSS exists by checking that it is an additive process \citep[i.e. that the triplet of the TSS  satisfies the conditions of][Th.9.8, p.52]{Sato} and that it is an additive subordinator (showing that it verifies the sufficient conditions of \textbf{Proposition \ref{p2.5}}).
	\begin{proposition}   \label{corollary:Additive_sub} 
		The additive tempered stable subordinator (TSS) exits and has $b_t=0$.
		\label{cor: 2}
	\end{proposition}
	\begin{proof}
		First, we prove that 	$ \left \{Z_t\right \}$ in \textbf{Definition \ref{definition:ATSSub}} is an additive process using  \citet[][Th.9.8, p.52]{Sato}; that is, we check whether the triplet introduced in (\ref{eq:TSS})  satisfies the theorem conditions.
		\begin{enumerate}
			\item The triplet has no diffusion term.
			\item ${	{\cal V}}_t$ is not decreasing in $t$ because both $\sigma_t^{2\alpha}t/ k_t^{1-\alpha}$ and $\sigma_t k_t$ are non-decreasing in $t$.
			\item For $t>0$, the continuity of ${ 	{\cal V}}_t(B)$, where $B\in\mathbb{B}\left(\mathbb{R}^+\right)$ and $B \subset \{ x : |x|>\epsilon >0 \}$, is due to the composition of  continuous functions. 
			For $t=0$ we can extend 
			${ 	{\cal V}}_t(B)$ and ${ \Gamma}_t$ to 0 since both converge to 0 as $t\rightarrow 0$. The convergence of ${ \Gamma}_t$ to 0 is due to ${ \Gamma}_t$ positiveness and  
			to 
			\begin{equation}
				\Gamma_t\leq	\int_{0}^{\infty}{\left(|x|\wedge 1\right){ 	{\cal V}_t\left({x} \right){dx}}}\leq	\int_{0}^{\infty}{x	{\cal V}_t\left(x \right)dx}=t\sigma^2_t\;\;.   \label{eq:conv} 
			\end{equation}
			
			The convergence of ${ 	{\cal V}}_t(B)$ to 0 is due to the dominated convergence theorem. We observe that, $ \forall x \in \mathbb{R}^+ $  such that
			$ |x|>\epsilon >0$, 
			${ {\cal V}}_t(x)$ is finite and a decreasing function of $t$.
		\end{enumerate}

		Second, we verify that the conditions of  {\bf Proposition \ref{p2.5}} on the generating triplet of an additive subordinator are satisfied 
		by $ \left \{ Z_t\right \}$. Let us observe that there is no diffusion term and,
		accordingly to (\ref{eq:conv}), $\int_{0}^{\infty}{\left(|x|\wedge 1\right){ 	{\cal V}_t\left({x} \right){dx}}}<\infty$. 
		
		Moreover, ${{\cal V}}_t((-\infty,0])=0$ and $b_t$ is null by direct substitution of $\Gamma_t$ in the formula of  {\bf Proposition \ref{p2.5}}. Thus, $ \left \{ Z_t\right \}$ is an additive subordinator\end{proof}
	In the next proposition, we utilize the TSS  to construct the ATS --with $\eta_t$ constant-- via additive subordination.
	\begin{proposition}\label{theore_if_only}
		It is possible to construct via additive subordination an additive process with characteristic function  (\ref{laplace}) if and only if $\eta_t$ is constant.
	\end{proposition}
	\begin{proof}
		We prove the only if using the results in  \textbf{Theorem \ref{theorem:no_subordination}}. In particular, we have shown 	that is impossible to construct via independent subordination the ATS, an additive process with characteristic function  (\ref{laplace}), if  $\eta_t$ is non-constant.
		
		We prove the if. Notice that  the conditions for the existence of the TSS are satisfied by any couple $k_t$, $\sigma_t$  that satisfies the condition of existence of the ATS \citep[cf.][Th.2.1, p.503]{azzone2019additive} with $\eta_t$ constant. Hence, for any ATS with $\eta_t$ constant, we can define the TSS.
		Thanks to \citet{mijatovic2010additive}, Prop.1, p.2,	the process $W(Z_t) +(1/2+\eta)Z_t$ is an additive process, where $\{Z_t\}$ is the TSS in (\ref{eq:TSS}) independent from the BM  and $\eta$ is a positive constant. 
		Adding the deterministic continuous function $\varphi_t t$ to $W(Z_t) +(1/2+\eta)Z_t$ preserves the additivity property.
		Moreover, thanks to the tower property of expectations, it is possible to verify  that $W(Z_t)-(1/2+{\eta})Z_t+\varphi_t t$ has the same characteristic function of the ATS
	\end{proof}

Let us comment on this  proposition. A subcase of the ATS --when $\eta_t$ is constant--  can be obtained through additive subordination. Interestingly, the parameter $\eta_t$ is connected with the model implied volatility skew; i.e. is the at-the-money derivative  of the implied volatility with respect to the strike price \citep[for a definition see e.g.,][Ch.3, p.35]{gatheral2011volatility}. 
It is well established, in the equity case, that the  short time skew is proportionally inverse to the square root of the time to maturity. 

In Figure \ref{Figure::Short_time_skew}, we present an example of  the short time implied volatility  skew for the EURO STOXX 50   at a given date,
the $21^{st}$ of March 2019.
We plot the market skew with respect to the time to maturity $t$: it appears to be well described by a fit $O\left(\sqrt{{1}/{t}}\right)$.  
\begin{center}
	\begin{minipage}[t]{1\textwidth}
		\includegraphics[width=\textwidth]{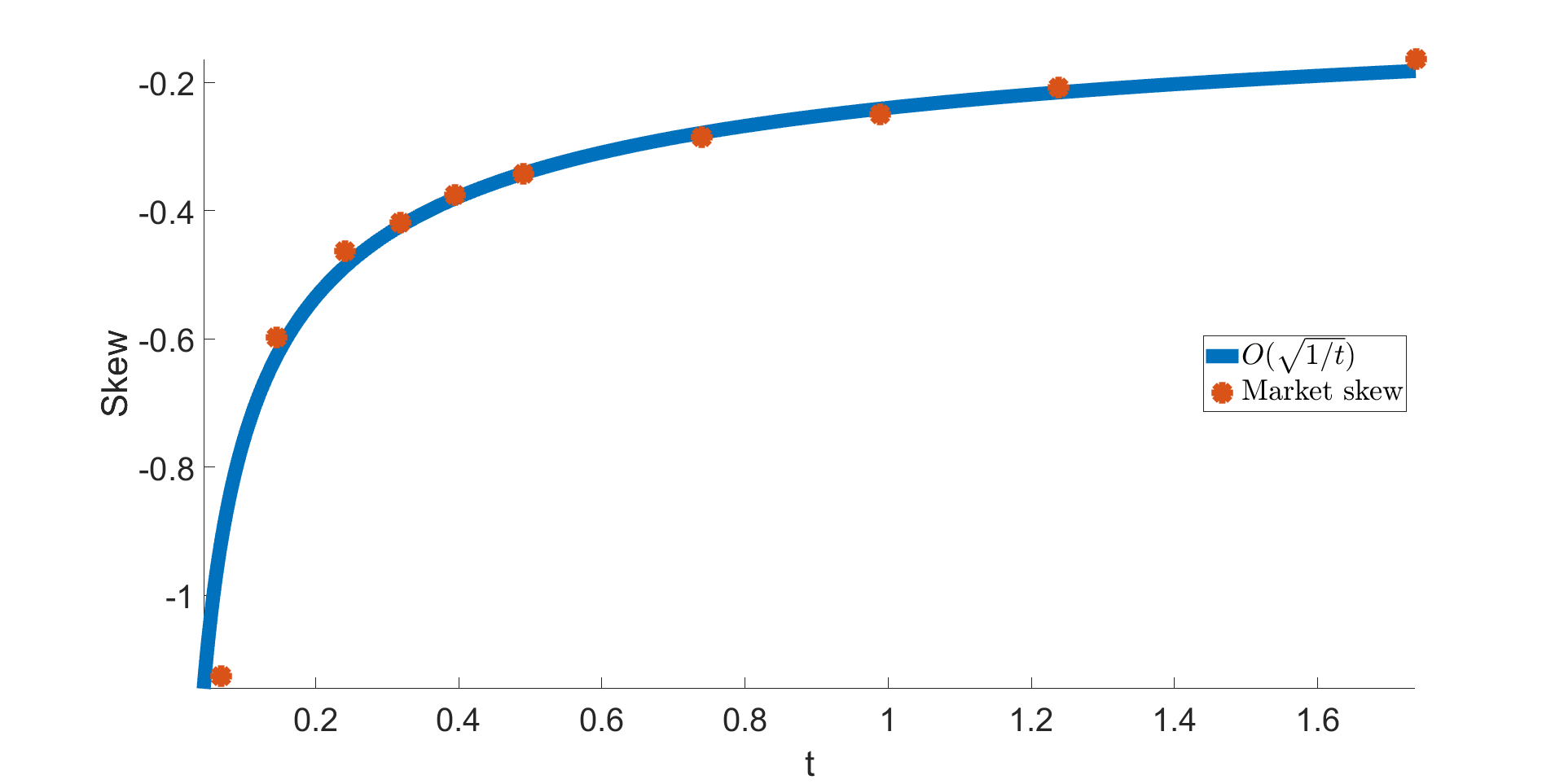}
		\label{Figure::Short_time_skew}\captionof{figure}{\small Example of the EURO STOXX 50 short-time implied volatility  skew on the $21^{st}$ of March 2019. We plot the market skew with respect to the time $t$ and an $O\left(\sqrt{{1}/{t}}\right)$ curve.}
	\end{minipage} 
\end{center}
It has been proven that the ATS reproduces this feature of the equity market implied volatility  if and only if $\eta_t$ is non constant and is proportional to $t^{-1/2}$ \citep{azzone2021short}. In the upcoming section, we propose a statistical test to provide evidence of a non-constant $\eta_t$ in the S\&P 500 and EURO STOXX 50 derivative markets.	
	\smallskip

	In \textbf{Proposition \ref{theore_if_only}}, we have established that 	it is possible to construct via additive subordination the ATS if and only if $\eta_t$ is constant. In the following, we demonstrate that it is not possible to construct the ATS with non-constat  $\eta_t$ even with a  generalization of  additive subordination: processes within the class $\{W(a_t  Z_t) + b_t  Z_t  + c_t\}$. Notice that this structure is more general than the subordinated BM (\ref{eq:sub}) because $a_t$ and $b_t$ are functions of time. The generalization stems naturally from L\'evy subordination. Let us explain why. The ATS is an extension of the well-known  L\'evy normal tempered stable (LTS) process \citep[see][Ch.4]{Cont}. The LTS $	\{g_t\} $ is usually built through L\'evy subordination as
	\begin{equation}
		g_t =  - \left( \eta + \frac{1}{2} \right) \; Z_t +  \; W_{Z_t}+ \varphi\, t\qquad \; ,
		\label{eq:LTS}    
	\end{equation}
	where 
	$ \eta$ is a real parameter 
	while 
	$\varphi$ is obtained by imposing the martingale condition on the forward price. 
	$\{Z_t\}$ is a L\'evy tempered stable subordinator with mean $\sigma^2t$ and variance $\sigma^4kt$. The most common example of LTS are NIG \citep{Barndorff1997NIG} and VG \citep{madan1998variance}. The former is built with an Inverse Gaussian L\'evy subordinator ($\alpha=1/2)$ while the latter with a Gamma L\'evy subordinator ($\alpha=0$). 
	
	\smallskip
	
	Can we build the ATS in a similar manner to  equation (\ref{eq:LTS}) but with time-dependent parameters and an additive subordinator? We consider the  process $	\{g_t\} $ such that
	\begin{equation}
		\label{eq: master_sub}
		g_t =-  \left( \eta_t + \frac{1}{2} \right) \; \;  Z_t +  \;  W_{ Z_t}+ \varphi_tt  \;\;, 
	\end{equation}
	where $\{Z_t\}$ is an additive subordinator. We point out that  $	\{g_t\}$ in (\ref{eq: master_sub}) is not a subordinated BM as in \eqref{eq:sub} because it   depends on $ \eta_t$ (a function of $t$). Let us notice that, if $\{Z_t\}$ is the TSS of \textbf{Definition \ref{definition:ATSSub}} the process $	\{g_t\}$ at time $t$ has the same characteristic function of the ATS \eqref{laplace}.
	Consequently, if a process in the form $\{W\left(a_t{Z}_t\right)+b_t {Z}_t+c_t\}$ is an additive process for  $a_t$ positive, increasing and continuous, and $b_t$ and $c_t$ real and continuous, we would obtain that the process in (\ref{eq: master_sub}) is identical in law to the ATS.
	However, in the next theorem, we demonstrate that  a process in the form $\{W\left(a_t{Z}_t\right)+b_t {Z}_t+c_t\}$ is an additive process if and only if $a_t$ and $b_t$ are constant; i.e. it corresponds to a subordinated BM \eqref{eq:sub}. 
	
	\begin{theorem} \label{theorem: sub_diff}
		The process $ \{W\left(a_t{Z}_t\right)+b_t {Z}_t+c_t\}$, with $a_t$ positive, increasing and continuous and $b_t$ real and continuous, is an additive process if and only if $a_t$ and $b_t$ are  constant in time.
	\end{theorem}
	\begin{proof}
		First, let us prove the only if part. We define $f_t:=W\left(a_t{Z}_t\right)+b_t {Z}_t+c_t$. We show that, if it exits $t>s$ such that $a_t\neq a_s$ or $b_t\neq b_s$ then the increment $f_t-f_s$ is not independent from $f_s$; thus, $\{f_t\}$ is not an additive process. 
		We use the fact that two random variables $Y_1$ and $Y_2$ are independent if and only if, for any $u_1,\, u_2\in \mathbb{C}$\begin{equation}
			\mathbb{E}\left[e^{iu_1Y_1+iu_2Y_2}\right]=\mathbb{E}\left[e^{iu_1Y_1}\right]\mathbb{E}\left[e^{iu_2Y_2}\right]\;\;.\label{eq:independence}
		\end{equation}
		We need to show that it exists a couple $u_1,\, u_2\in \mathbb{C}$ such that \begin{equation}
			\mathbb{E}\left[e^{iu_1(f_t-f_s)+iu_2f_s}\right]\neq\mathbb{E}\left[e^{iu_1(f_t-f_s)}\right]\mathbb{E}\left[e^{iu_2f_s}\right]\;\;. \label{eq:independence_incr}
		\end{equation}
		The left hand side of \eqref{eq:independence_incr} \begin{align}
			\mathbb{E}\left[e^{iu_1(f_t-f_s)+iu_2f_s}\right]=\nonumber&\mathbb{E}\left[e^{iu_1(b_tZ_t-b_sZ_s)-\frac{1}{2}u_1^2(a_tZ_t-a_sZ_s)+iu_2b_sZ_s-\frac{1}{2}u_2^2(a_sZ_s)}\right]\\
			=&\nonumber\mathbb{E}\left[e^{iu_1b_t(Z_t-Z_s)-\frac{1}{2}u_1^2a_t(Z_t-Z_s)}\right]\\&\quad \mathbb{E}\left[e^{iu_1(b_t-b_s)Z_s-\frac{1}{2}u_1^2(a_t-a_s)Z_s+iu_2b_sZ_s-\frac{1}{2}u_2^2(a_sZ_s)}\right] \nonumber\\
			=&\mathbb{E}\left[e^{iu_1b_t(Z_t-Z_s)-\frac{1}{2}u_1^2a_t(Z_t-Z_s)}\right] \mathbb{E}\left[e^{i\hat{u}_1Z_s+i\hat{ u}_2Z_s}\right]\;\;,\label{eq:left}
		\end{align}
		where $i\hat{u}_1:=iu_1(b_t-b_s)-\frac{1}{2}u_1^2(a_t-a_s)$ and  $i\hat{u}_2:=iu_2b_s-\frac{1}{2}u_2^2(a_s)$.
		We obtain the first equality by tower property of the expected value by conditioning with respect to the filtration of $\{Z_t\}$ and because $a_sZ_s$ is non-decreasing. The second equality follows by the fact that the increment $Z_t-Z_s$ is independent from $Z_s$.
		
		The right hand side of \eqref{eq:independence_incr}
		
		\begin{align}
			\mathbb{E}\left[e^{iu_1(f_t-f_s)}\right]\mathbb{E}\left[e^{iu_2f_s}\right]&=\mathbb{E}\left[e^{iu_1(b_tZ_t-b_sZ_s)-\frac{1}{2}u_1^2(a_tZ_t-a_sZ_s)}\right]\mathbb{E}\left[e^{i\hat{u}_2Z_s}\right]\nonumber \\&=\mathbb{E}\left[e^{iu_1b_t(Z_t-Z_s)-\frac{1}{2}u_1^2a_t(Z_t-Z_s)}\right] \mathbb{E}\left[e^{i\hat{u}_1Z_s}\right]\mathbb{E}\left[e^{i\hat{u}_2Z_s}\right]\label{eq:right}\;\;.
		\end{align}
	As above, the second equality follows by the fact that the increment $Z_t-Z_s$ is independent from $Z_s$.
	
	We show that \eqref{eq:independence_incr} is true by contradiction. Let us assume that \eqref{eq:independence_incr} holds with the equality:
	\[
		\mathbb{E}\left[e^{iu_1(f_t-f_s)+iu_2f_s}\right]=\mathbb{E}\left[e^{iu_1(f_t-f_s)}\right]\mathbb{E}\left[e^{iu_2f_s}\right]\;\;. 
	\]
		By substituting \eqref{eq:left} and \eqref{eq:right} we obtain
		\begin{equation}
			\mathbb{E}\left[e^{(i\hat{u}_1+i\hat{ u}_2)Z_s}\right]=\mathbb{E}\left[e^{i\hat{u}_1Z_s}\right]\mathbb{E}\left[e^{i\hat{u}_2Z_s}\right]\;\;,
		\end{equation}
	where the process $\{Z_t\}$ appears on both sides of the equality.
		The last equation is absurd for any couple of $\hat{u}_1$ and $\hat{u}_2$  because $\{Z_t\}$ is not deterministic. This proves the only if part.
		
		We prove the if part using Proposition 1 of \citet{mijatovic2010additive}. If $a_t={a}$ and $b_t={b}$ are  constant in time $W({a}t)+{b}t$ is a L\'evy process with diffusion term ${a}$ and drift term ${b}$. Then, $\{W({a}Z_t)+{b}Z_t\}$ is an additive process
	\end{proof}
	
	In this section, we have proven two main results regarding additive subordination and its limitations. First, in \textbf{Theorem \ref{theorem:no_subordination}}, we establish that an additive process is not always identical in law to a subordinated BM whatever the choice of the subordinator.  Second, in \textbf{Theorem \ref{theorem: sub_diff}},  we demonstrate that  also considering a natural extension of additive subordination we cannot construct a process identical in law to an additive process.
	These results have been proven considering as a counterexample the class of ATS processes with non-constant $\eta_t$. In the next section, we provide statistical evidence, on a large dataset of options, that only a non-constant $\eta_t$ is  consistent with market prices. In light of these findings, we can tentatively answer  the question posed in the title. For a large class of additive processes, subordination (i.e. the case with constant $\eta_t$) is irrelevant for applications in the derivative market.

	\section{Dataset \& Numerical experiments} \label{section:num}
	In this section, we present the dataset, and propose a statistical test to verify whether the ATS with constant $\eta_t$ can calibrate the equity derivatives' implied volatility surface. In the previous section, we have proven that this is the unique subcase of the ATS which can be constructed through additive subordination. We find, in all business days in a six-month time horizon,  strong statistical evidence to reject the null hypothesis  of a constant $\eta_t$;  hence,  additive subordination cannot  describe adequately the equity market prices.
	
	\smallskip
	We consider all  closing prices for  S\&P 500  and EURO STOXX 50 options\footnote{ We consider the CBOE European options on the  S\&P 500   and EUREX European options for the EURO STOXX 50. The Eikon Reuters option chains are  \textit{0\#SPX*.U} for the S\&P 500 and \textit{0\#STXE*.EX} for the EURO STOXX 50.} 
	each business day in the first semester of 2019 from the $1^{st}$ of January 2019 to the $30^{th}$ of June 2019. For every business day, we observe option prices not only on several strikes but also on several expiries:
	the third Friday of the first six months after the considered business day and then on  
	March, June, September, and December in the front year and June and December in the next year. For the EURO STOXX 50, options expiring in December of the following years are also available. Financial data are provided by Eikon Reuters. 
	We  obtain  interest rates and future prices following the methodology introduced in \citet{azzone2021funding}. Moreover, we consider options within   the (10\%, 90\%) Black and Scholes delta interval.
	
	\smallskip
	
	A brief description of the numerical methodology follows. We consider the ATS NIG ($\alpha=1/2$) and the ATS VG ($\alpha=0$). Similar results hold for any choice of $\alpha \in [0,1)$. The methodology is divided into two steps.
	
	First, for every business day, the ATS (NIG or VG) parameters are calibrated following the same procedure described in \citet[][Ch.14, pp.464-465]{Cont}: we consider the options maturity by maturity and obtain, for every maturity $T$, the three time-dependent parameters $k_T$, $\sigma_T$ and $\eta_T$.  
	
	Second, we define $\theta:=T\sigma_T^2$ and the ATS process  $\hat{ f}_{\theta}$ with parameters $\hat{k}_\theta:=k_T \sigma_T^2$, $\hat{\eta}_\theta:=\eta_T$ and $\hat{\sigma}_\theta:=1$. This new process, which is defined with respect to the new time $\theta$, has the same characteristic function of the calibrated ATS for every maturity.  Then, testing the null hypothesis that $\hat{\eta}_\theta$ is constant with respect to time $\theta$ is equivalent to testing that ${\eta}_T$ is constant with respect to the original time $T$. When considering a power scaling behavior of $\hat{\eta}_\theta$ 
	\begin{equation}
		\hat{\eta}_\theta =\;{\eta}\theta^\delta \;\;,
		\label{eq:scaling}
	\end{equation}
	where ${\eta}$ is positive  and  $\delta$ is real. In all days analyzed, this self-similar 	$\hat{\eta}_\theta$ appears to describe accurately market data for all available maturities --from a few days up to some years-- for both markets and both the NIG and VG ATS. 
	In figure 	\ref{figure:eta_scaling}, we plot a representative  example. We consider the VG ATS calibrated values of 	$\hat{\eta}_\theta$ in log-log scale, for the S\&P 500  on the $21^{st}$ of March 2019. We also report confidence intervals and a 	weighted regression line. 
	The methodology is described in detail by \citet{azzone2019additive}, appendix B. We estimate  $\delta$ as the slope coefficient of the regression and we  test whether it is non-zero. Notice that, if $\delta\neq 0$, $\hat{\eta}_\theta$ is non-constant. 
	\begin{center}
		\begin{minipage}[t]{1\textwidth}
			\centering
			\includegraphics[width=0.9\textwidth]{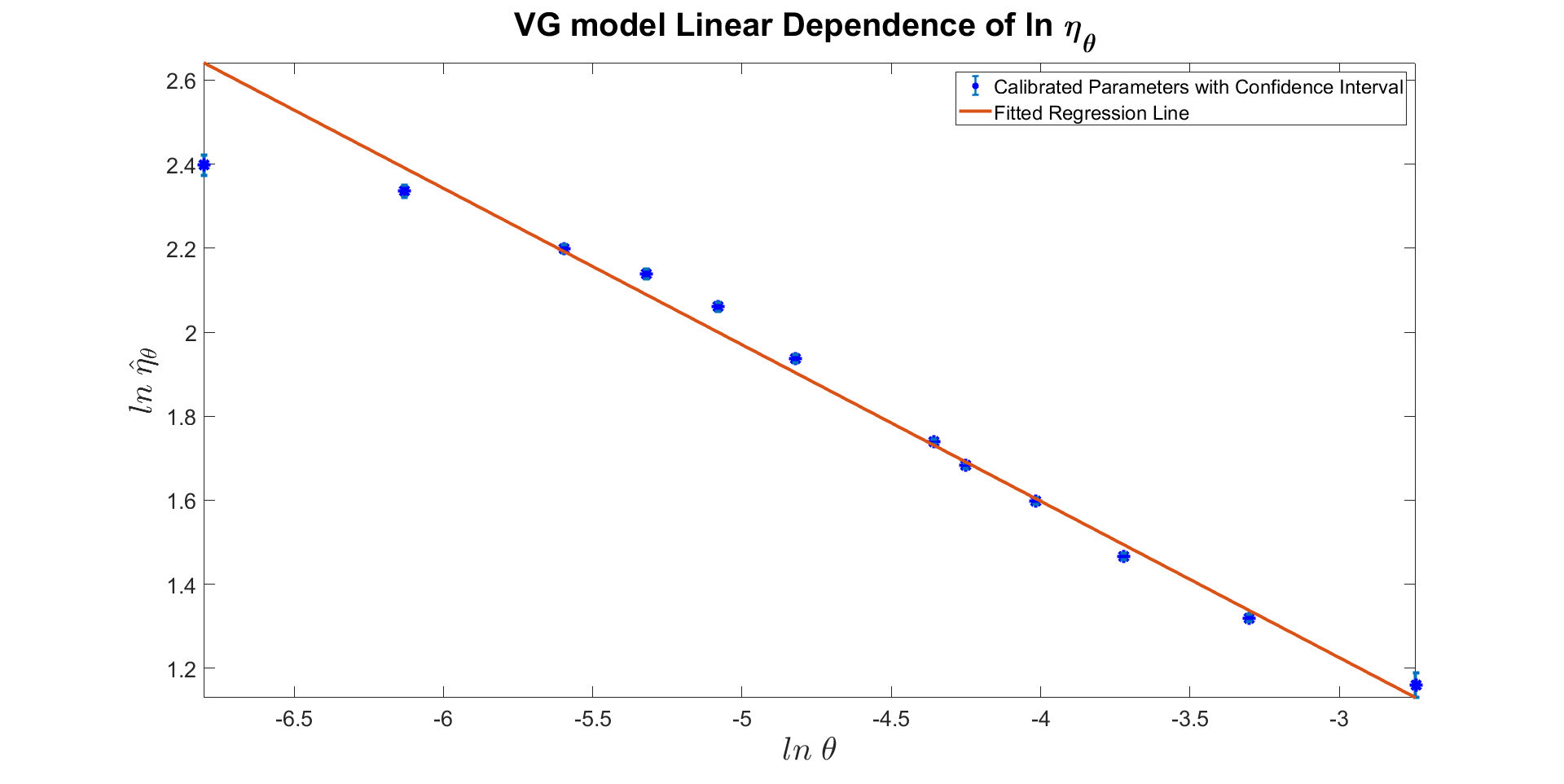} 
			\captionof{figure}{\small 
				Weighted regression line and  observed  $\ln {\hat \eta}_\theta$ with respect to $\ln\theta$ for the VG calibrated model on the $21^{st}$ of March 2019 for the S\&P 500 index. We plot confidence intervals equal to two times the corresponding standard deviations. In the analysis, we consider all available times to maturity that span from a few days to almost three years. Similar results hold for both indexes in all business days of the observed period.}\label{figure:eta_scaling}
		\end{minipage}
	\end{center}
	
	We test whether  $\delta$ is non-zero for the NIG and the VG ATS in the six-month dataset described at the beginning of this section: we calibrate the ATS on option prices for each day of the dataset and repeat the statistical test on $\delta$ each day.  The mean and maximum p-values of the statistical tests are reported in table \ref{tab: p_values}. 
	
	In all cases, we reject the null hypothesis of  $\delta=0$ with a 1$\tcperthousand$ threshold. Hence, for both indexes and both models (ATS VG and NIG), and in all considered days of the semester, there is statistical evidence that $\hat{\eta}_\theta$ is non-constant. 
	\begin{center}
		\begin{tabular} {|cc|cc|}
			\toprule
			Model	&Index	&Mean&Maximum\\
			\hline
			NIG&S\&P 500& $5*10^{-6}$ &$5*10^{-4}$\\
			\hline
			VG&S\&P 500&$10^{-152}$ &$10^{-150}$\\
			\hline
			NIG&EURO STOXX 50& $2*10^{-19}$ &$2*10^{-17}$\\
			\hline
			VG&EURO STOXX 50&$3*10^{-6}$ &$4*10^{-4}$\\
			\bottomrule
		\end{tabular}		
		\captionof{table}{\small Mean and maximum 
			values of the statistical tests. We refuse the null hypothesis of $\delta=0$ for the NIG and the VG model for every business day from  the $1^{st}$ of January 2019 to the $30^{th}$ of June 2019.
		}		\label{tab: p_values}
	\end{center}
	
	From the results in the table, we reject the null hypothesis with a 1$\tcperthousand$
	 threshold in all cases. Let us also underline that all p-values, in all days and for both models and indexes, are  below $10^{-4}$ with two exceptions. \\In appendix \ref{app:calib}, we provide  additional calibration result that shows why the ATS with constant $\eta_t$ is not consistent with market data. We compare the calibration performances of an ATS and an ATS with constant $\eta_t$. While for long maturity both reproduce accurately the smile, for shorter maturities the latter deviates considerably. This fact is in line with the theoretical results on the short time implied volatility skew: the ATS reproduces the power scaling skew if and only if $\eta_t$ is non constant.
	
	\bigskip
	
	%
	\section{Conclusions} \label{sec:conc}
In financial applications, the independence between the BM and the time change is a standard assumption.
This paper presents three interesting results on independent time changes of Brownian motions.

\smallskip
First, our pivotal observation is that a local semimartingale, in general, is not equivalent to a BM with a time-change \textit{independent} from the BM. In \textbf{Theorem \ref{theore_if_only}}, we have utilized  the additive normal tempered stable process (ATS) --a class of processes that calibrates accurately equity option markets-- as an example of additive process that cannot be represented with an independent time-change. This class of processes presents a deterministic function of time $\eta_t$ that plays a key role in reproducing accurately the equity volatility surface.
	Specifically, we have established the nonexistence of any pair comprising  a time change and an independent BM, where the resulting time-changed BM is equal in law to an ATS with a non-constant ${\eta}_t$.
	Hence, a representation theorem equivalent to the one of	\citet{monroe1978processes}  does not hold when we assume independent subordination.  
		
	\smallskip
		Second, we have established  in \textbf{Theorem \ref{theorem: sub_diff}} that an additive process  cannot even be expressed as the natural generalization  of additive subordination:  a process in the form $\{W(a_t  Z_t) + b_t  Z_t  + c_t\}$, where $a_t$ and $b_t$ are non-constant functions of time and $Z_t$ is an additive subordinator. Once again, we have utilized the ATS as a counterexample to prove the theorem. 
 Furthermore, we prove that only a specific subcase of the ATS, characterized by a constant  $\eta_t$, can be obtained through independent (additive) subordination (\textbf{Proposition \ref{theore_if_only}}). 
 
 \smallskip
Thus, a natural question arises for a practitioner's perspective. Can the aforementioned subcase of the ATS adequately replicate market data? We devise a statistical test to address this question. We have provided strong statistical evidence that the ATS  with a constant  parameter $\eta_t$ is inconsistent with the two most liquid equity derivative markets; we have considered all closing prices for S\&P 500 and EURO STOXX 50 options each business day in the first semester of 2019. We report the results  in table  \ref{tab: p_values}: for every day within this period, we reject the null hypothesis of a constant $\eta_t$ with a 1$\tcperthousand$ threshold.
Thus, independent subordination does not adequately replicate market data.

			\clearpage
	\bibliography{sources}
	\bibliographystyle{tandfx}

\section*{Acknowledgments}
		We thank G. Callegaro, G. Guatteri, E. Scalas, L. Torricelli  for their useful comments and P. Carr for enlightening discussions on the topic.
\clearpage
\appendix
\section{Calibration Results} \label{app:calib}
In this appendix, we report some additional calibration results to highlight  that an ATS with constant $\eta_t$ lacks the capacity to calibrate adequately market data; this is the unique subcase of ATS that can be represented via additive subordination (see \textbf{Proposition \ref{theore_if_only}}). In the following, we utilize the dataset described in section \ref{section:num}: for every day in a six-month time horizon,	 all  closing prices of  S\&P 500  and EURO STOXX 50 options within the range (10\%, 90\%) of Black and Scholes delta.
\smallskip

We calibrate the ATS and the ATS with constant $\eta_t$ on the same implied volatility surface and compare the results. 
In figure \ref{figure:SP_IV}, we report the volatility smile reproduced by the ATS NIG  on the $21^{st}$ of March 2019 on the S\&P 500 implied volatility surface.  The calibrated ATS implied volatility (blue dots) closely aligns with the market smile (green diamonds) both on the one month (on the left) and the one year (on the right) maturity. In contrast, the ATS with constant $\eta_t$ (red triangles) conspicuously fails to capture both the level and the skew of the short time implied volatility. In figure \ref{figure:EU_IV}, we present analogous findings for the EURO STOXX 50 volatility surface when calibrating the ATS VG. Once again, it is evident that the ATS with constant $\eta_t$ fails to replicate market data. 
\begin{center}
	\begin{minipage}[t]{1\textwidth}
		\centering
		\includegraphics[width=0.9\textwidth]{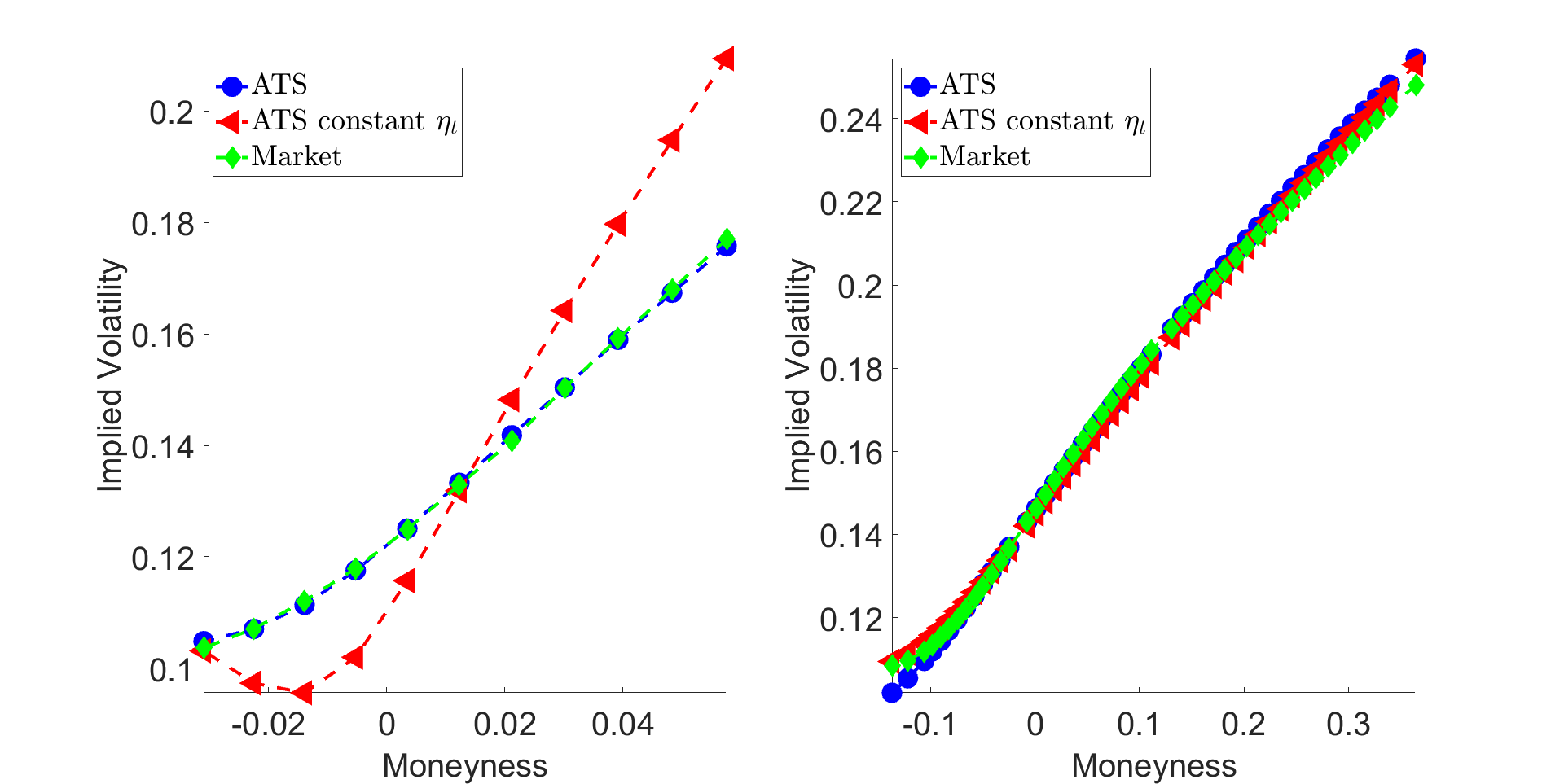} 
		\captionof{figure}{\small ATS NIG (blue dots), ATS NIG with constant $\eta_t$ (red triangles) and market (green diamond) implied volatility on the S\&P 500 surface on the $21^{st}$ of March 2019. The calibrated ATS fits very well the market implied volatility both on the one month (on the left) and the one year (on the right) maturity  while the ATS with constant $\eta_t$ is unable to match the short time smile.
		}\label{figure:SP_IV}
	\end{minipage}
\end{center}
\begin{center}
	\begin{minipage}[t]{1\textwidth}
		\centering
		\includegraphics[width=0.9\textwidth]{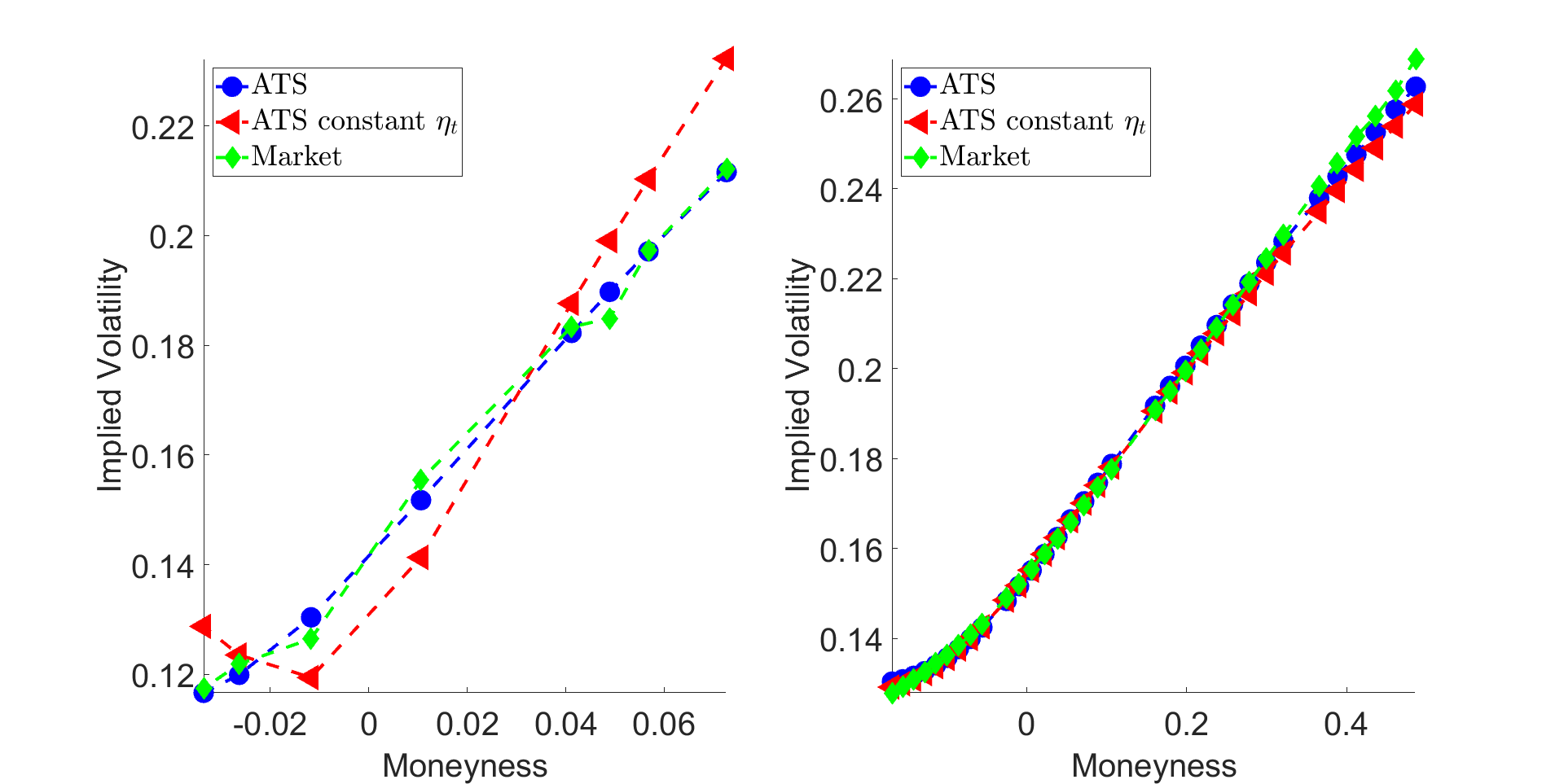} 
		\captionof{figure}{\small ATS VG (blue dots), ATS VG with constant $\eta_t$ (red triangles) and market (green diamond) implied volatility on the EURO STOXX 50 surface on the $21^{st}$ of March 2019. The  ATS fits very well the market implied volatility both on the 1 month (on the left) and the one year (on the right) maturity.
	 As in the case of the S\&P 500 the ATS with constant $\eta_t$ is unable to match the short time smile.	}\label{figure:EU_IV}
	\end{minipage}
\end{center}
In figure \ref{figure:MSE}, we report the ATS NIG (blue dots) and ATS NIG with constant $\eta_t$ (red triangles) mean squared error (MSE) with respect to different times to maturity in log-scale. While for long maturities the MSE are similar,  for short maturities the ATS with constant $\eta_t$ has a MSE more than two orders of magnitude above the ATS.

\begin{center}
	\begin{minipage}[t]{1\textwidth}
		\centering
		\includegraphics[width=0.9\textwidth]{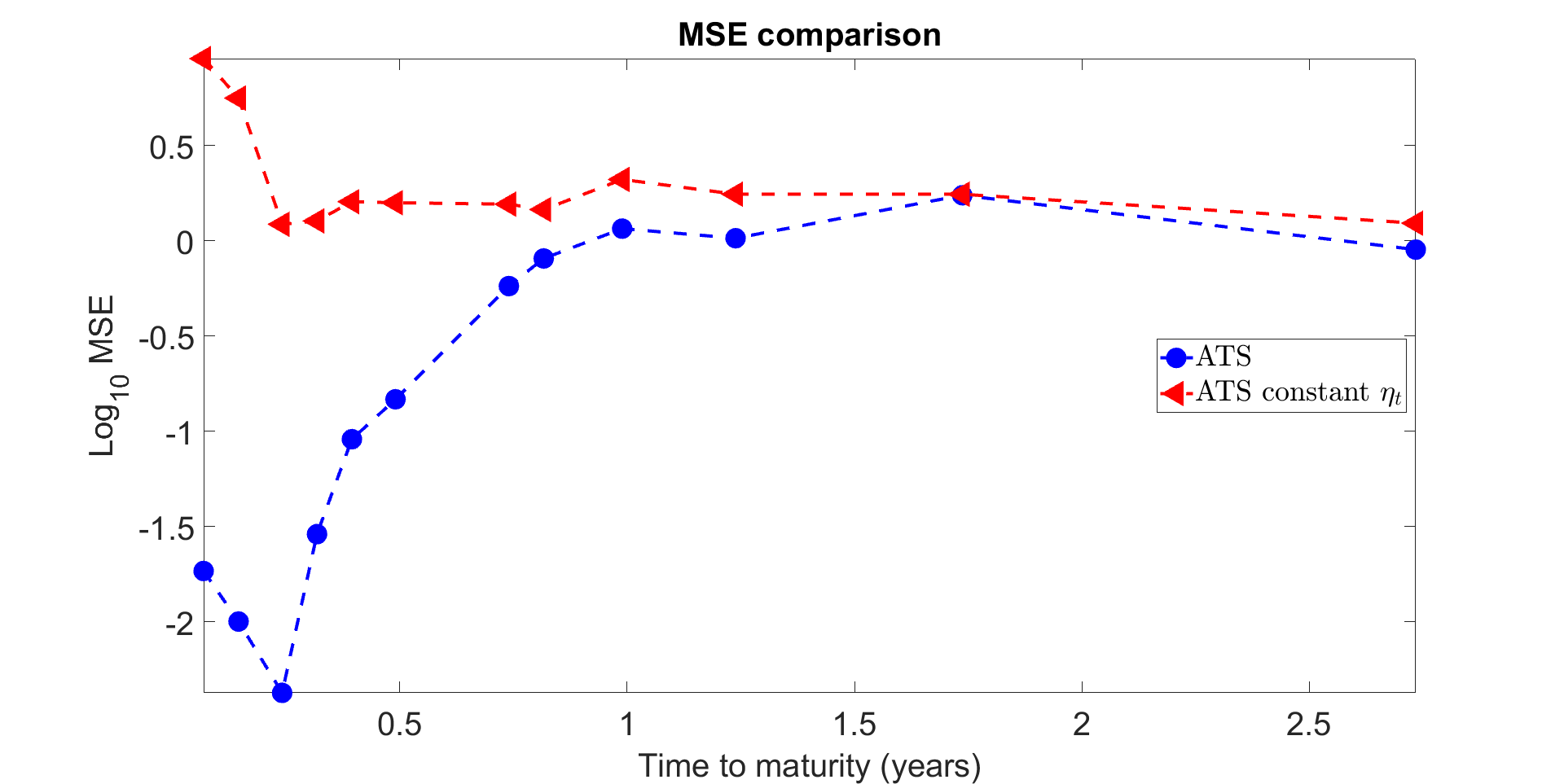} 
		\captionof{figure}{\small ATS NIG (blue dots) and ATS NIG with constant $\eta_t$ (red triangles) mean squared error (MSE) with respect to different times to maturity in log-scale on the $21^{st}$ of March 2019. Let us notice that, while for long maturities the MSE are similar,  for short maturities the ATS with constant $\eta_t$ has a MSE more than two orders of magnitude above the ATS.	}\label{figure:MSE}
	\end{minipage}
\end{center}
These results are in line with the discussion in section \ref{sec:main}: the theoretical skew of the ATS matches the short time power scaling behavior observed in the market if and only $\eta_t$ is not constant. We can understand from figures \ref{figure:SP_IV}-\ref{figure:MSE}  how much this constraint on $\eta_t$ affects the capability of the ATS to match the market implied volatility for short maturities.
\smallskip

In figure \ref{figure:Eta_Historic_EU}, we report the calibrated $\delta$ for the ATS NIG (on the left) and the ATS VG (on the right) for every day in the semester for the EURO STOXX 50 surface. We recall that $\delta=0$ is equivalent to a constant $\eta_t$ (see equation (\ref{eq:scaling})). In both cases, $\delta$ is far from zero. 
	\begin{center}
	\begin{minipage}[t]{1\textwidth}
		\centering
		\includegraphics[width=0.9\textwidth]{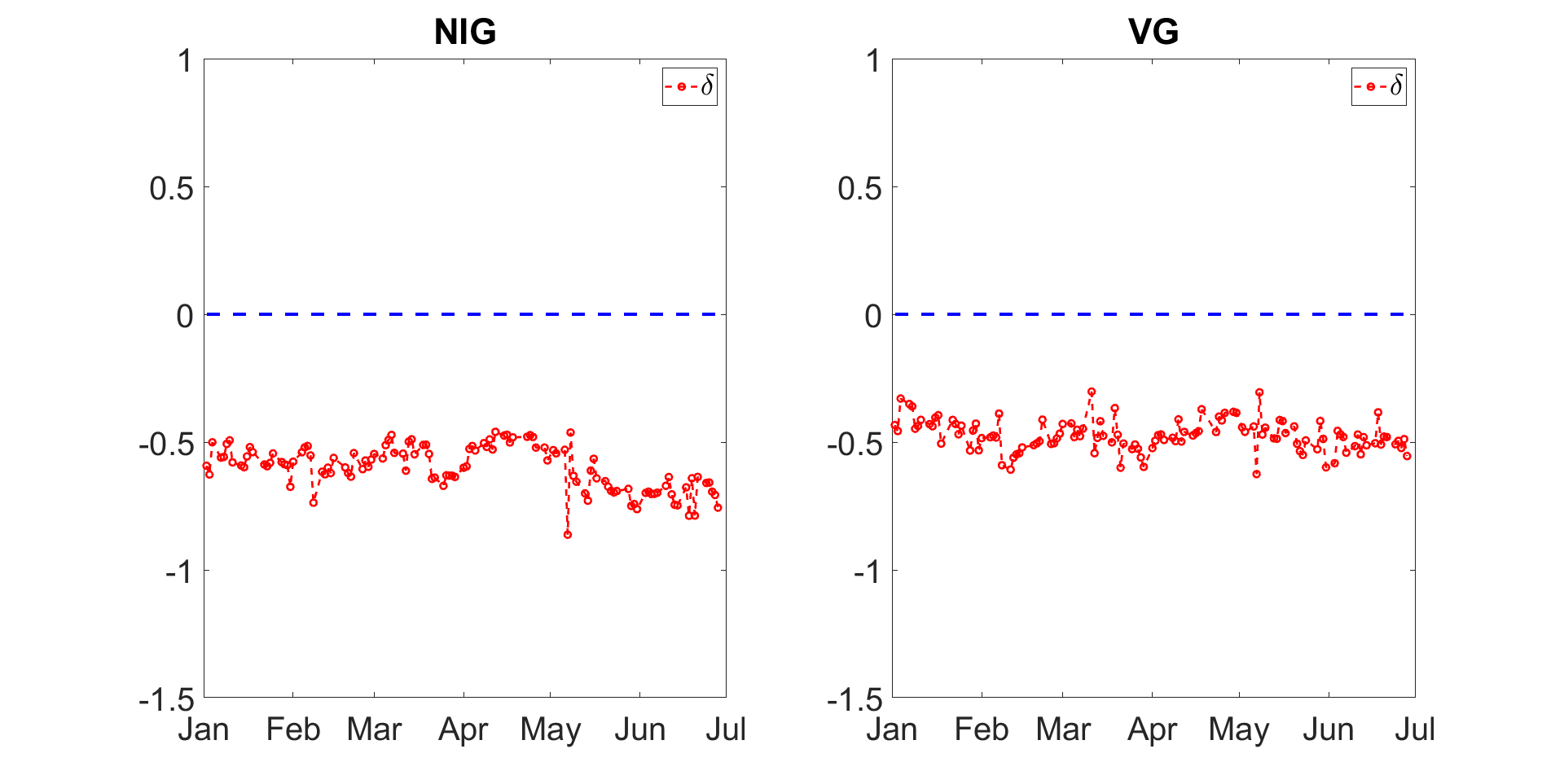} 
		\captionof{figure}{\small  Estimated $\delta$ for the ATS NIG (on the left) and the  VG (on the right) for all business days of the first semester of 2019 for the EURO STOXX 50 index. We observe that in all business days and for both models the estimated $\delta$ (red dot)  is far from zero (the dashed blue line). 
		}\label{figure:Eta_Historic_EU}
	\end{minipage}
\end{center}
In section \ref{section:num}, we have performed a statistical test with null hypothesis $\delta=0$ for all business days of the semester. In figure \ref{figure:p_values}, we report the p-values of the statistical tests for the S\&P 500 index.  Let us notice that all p-values are  below $0.1\tcperthousand$ with two exceptions in the NIG case and no exceptions in the VG case.
	\begin{center}
	\begin{minipage}[t]{1\textwidth}
		\centering
		\includegraphics[width=0.9\textwidth]{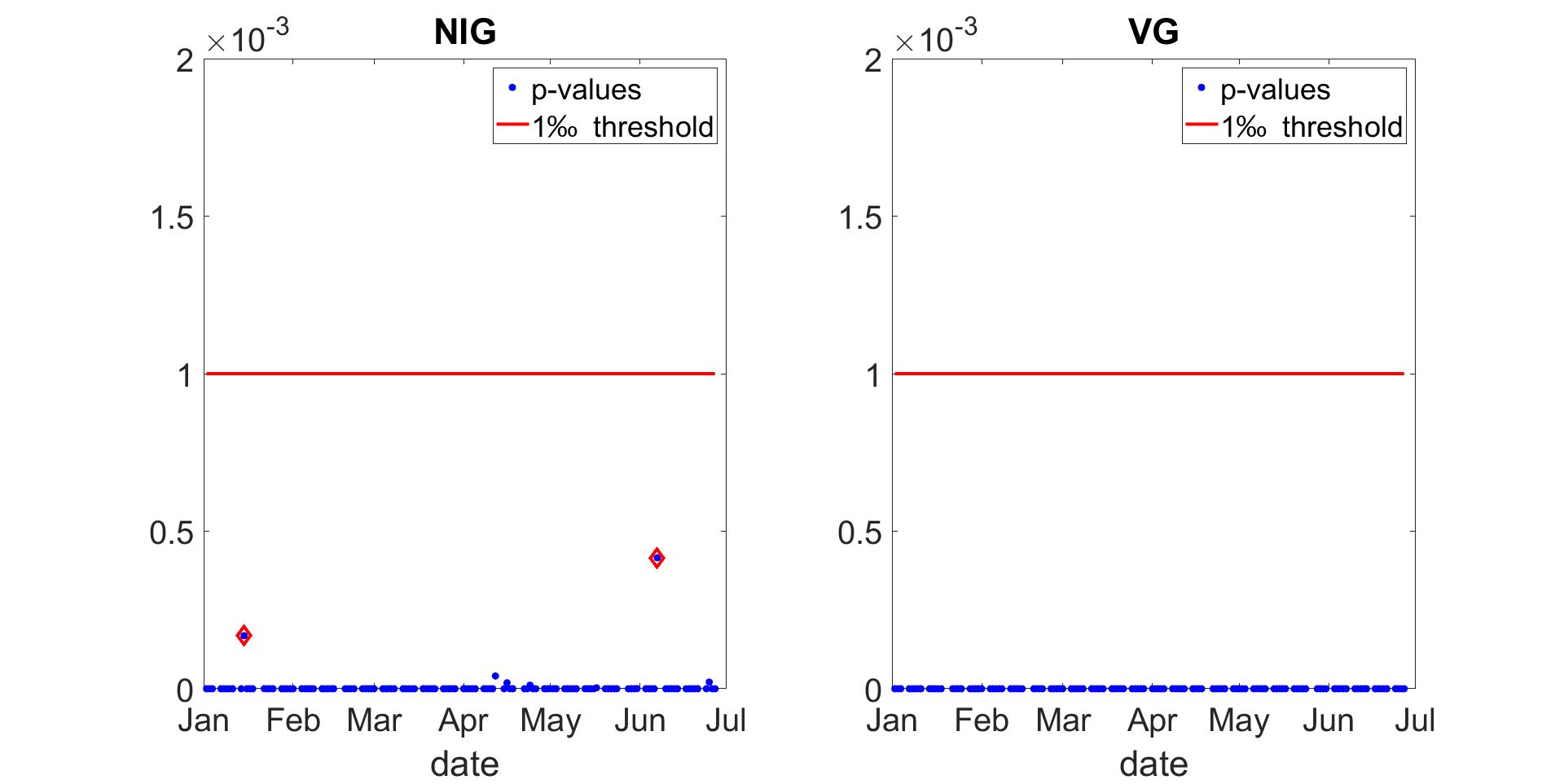} 
		\captionof{figure}{\small  	P-values (blue dots) of the statistical tests with null hypothesis $\delta=0$ for the ATS NIG (on the left) and the  VG (on the right) for all business days of the semester for the S\&P 500 index. The continuous red line is the 1$\tcperthousand$  threshold of the test. In all business days, and for both models, we reject the null hypothesis with a 1$\tcperthousand$ threshold. Let us notice that all p-values are  below 0.1$\tcperthousand$ with two exceptions (identified with the red diamonds) in the NIG case and no exceptions in the VG case.
		}\label{figure:p_values}
	\end{minipage}
\end{center}
\end{document}